\documentclass[conference]{IEEEtran}
\IEEEoverridecommandlockouts
\usepackage{cite}
\usepackage{amsmath,amssymb,amsfonts}
\usepackage{algpseudocode}
\usepackage{algorithm}
\usepackage{graphicx}
\usepackage{textcomp}
\usepackage{amsmath}
\usepackage{amssymb}
\usepackage{mathtools}
\usepackage{amsthm}
\usepackage{xcolor}

\theoremstyle{plain}
\newtheorem{theorem}{Theorem}[section]
\newtheorem{proposition}[theorem]{Proposition}

\theoremstyle{definition}
\newtheorem{definition}[theorem]{Definition}

\theoremstyle{remark}

\def\BibTeX{{\rm B\kern-.05em{\sc i\kern-.025em b}\kern-.08em
    T\kern-.1667em\lower.7ex\hbox{E}\kern-.125emX}}
\begin{document}

\title{Reducing QAOA Circuit Depth by Factoring out Semi-Symmetries}

\author{
\IEEEauthorblockN{Jonas Nüßlein}
\IEEEauthorblockA{\textit{Institute of Computer Science} \\
\textit{LMU Munich}\\
Germany \\
jonas.nuesslein@ifi.lmu.de}
\and
\IEEEauthorblockN{Leo Sünkel}
\IEEEauthorblockA{\textit{Institute of Computer Science} \\
\textit{LMU Munich}\\
Germany}
\and
\IEEEauthorblockN{Jonas Stein}
\IEEEauthorblockA{\textit{Institute of Computer Science} \\
\textit{LMU Munich}\\
Germany}
\and
\IEEEauthorblockN{Tobias Rohe}
\IEEEauthorblockA{\textit{Institute of Computer Science} \\
\textit{LMU Munich}\\
Germany}
\and
\IEEEauthorblockN{Daniëlle Schuman}
\IEEEauthorblockA{\textit{Institute of Computer Science} \\
\textit{LMU Munich}\\
Germany}
\and
\IEEEauthorblockN{Claudia Linnhoff-Popien}
\IEEEauthorblockA{\textit{Institute of Computer Science} \\
\textit{LMU Munich}\\
Germany}
\and
\IEEEauthorblockN{Sebastian Feld}
\IEEEauthorblockA{\textit{Quantum \& Computer Engineering Department} \\
\textit{Delft University of Technology}\\
The Netherlands}
}

\maketitle

\begin{abstract}

QAOA is a quantum algorithm for solving combinatorial optimization problems. It is capable of searching for the minimizing solution vector $x$ of a QUBO problem $x^TQx$. The number of two-qubit CNOT gates in the QAOA circuit scales linearly in the number of non-zero couplings of $Q$ and the depth of the circuit scales accordingly. Since CNOT operations have high error rates it is crucial to develop algorithms for reducing their number. We, therefore, present the concept of \textit{semi-symmetries} in QUBO matrices and an algorithm for identifying and factoring them out into ancilla qubits. \textit{Semi-symmetries} are prevalent in QUBO matrices of many well-known optimization problems like \textit{Maximum Clique}, \textit{Hamilton Cycles}, \textit{Graph Coloring}, \textit{Vertex Cover} and \textit{Graph Isomorphism}, among others. We theoretically show that our modified QUBO matrix $Q_{mod}$ describes the same energy spectrum as the original $Q$. Experiments conducted on the five optimization problems mentioned above demonstrate that our algorithm achieved reductions in the number of couplings by up to $49\%$ and in circuit depth by up to $41\%$.

\end{abstract}


\section{Introduction}

The Quantum Approximate Optimization Algorithm (QAOA) \cite{farhi2014quantum} is designed to tackle combinatorial optimization problems using quantum computers by preparing a quantum state that maximizes the expectation value of the cost-hamiltonian. QAOA is widely recognized as a prime contender for showcasing quantum advantage on Noisy Intermediate-Scale Quantum (NISQ) devices \cite{zou2023multiscale}. It aims to approximate the ground state of a given physical system, often referred to as the Hamiltonian. However, its successful implementation faces challenges due to the high error rates inherent in current near-term quantum devices, which lack full error correction capabilities.

Utilizing QAOA to solve a problem entails a two-step process. Initially, the problem is translated into a parametric quantum circuit consisting of $p$ layers each consisting of $2$ adjustable parameters, where $p$ is a hyperparameter that needs to be set manually. This circuit is then run for thousands of trials. Subsequently, a classical optimizer utilizes the expectation value of the output distribution to refine the parameters. This iterative process continues until the optimal parameters for the circuit are determined. The cost function, which QAOA tries to minimize is usually represented as a Quadratic Unconstrained Binary Optimization (QUBO) problem.

The quantity of two-qubit CNOT operations within a QAOA circuit is equal to $2C \cdot p$ where $C$ is the number of non-zero couplings in the QUBO (number of edges in the problem graph). However, CNOT operations are susceptible to errors and often lead to prolonged runtimes. For instance, on the Google Sycamore quantum processor \cite{ayanzadeh2023frozenqubits}, CNOTs exhibit an average error-rate of $1\%$. Furthermore, CNOT gates might require additional SWAP gates since control- and target-qubit might not be connected on the hardware chip. Therefore, minimizing the number of these operations becomes crucial to improve the efficiency and accuracy of quantum optimization algorithms like QAOA.

In this paper, we therefore propose a method for using ancilla qubits to reduce the number of non-zero couplings and therefore also the number of CNOT operations and the depth of the QAOA circuit. We do this by identifying \textit{semi-symmetries} (defined in \textbf{III.4}) which we factor out into ancilla qubits. Our algorithm can therefore create different QAOA circuits for the same QUBO by trading-off the number of qubits and the number of non-zero couplings (which translates linearly to the number of CNOT gates). To demonstrate the effectiveness of our approach, we tested it on five well-known optimization problems: Maximum Clique, Hamilton Cycles, Graph Coloring, Vertex Cover, and Graph Isomorphism. Our results show that our method can reduce the number of couplings by up to $49\%$ and the circuit depth by up to $41\%$, thus significantly improving the efficiency and scalability of QAOA for solving a diverse range of NP-hard optimization problems. Through our approach, we aim to contribute to the advancement of quantum optimization algorithms and facilitate their practical deployment in solving real-world optimization challenges.

\section{Background}

\subsection{Quadratic Unconstrained Binary Optimization}
Let $Q$ be a symmetric, real-valued $(n \times n)$-matrix and $x \in \mathbb{B}^n$ be a binary vector. Quadratic Unconstrained Binary Optimization (QUBO) ~\cite{10} is an optimization problem of the form:
\begin{equation}
x^* = \underset{x}{argmin} \: H(x) = \underset{x}{argmin} \: \sum_{i=1}^{n}\sum_{j=i}^{n}{x_i\ x_j\ Q_{ij}}
\end{equation}

The function $H(x)$ is usually called \textit{Hamiltonian}.
We will refer to the matrix $Q$ as the ``QUBO matrix'' in this paper.

The optimization task is to find a binary vector $x$ that is as close as possible to the optimum. Finding the optimum $x^*$ is known to be \textit{NP}-hard \cite{glover2018tutorial}. QUBOs attracted special attention recently since they can be solved using Quantum Optimization approaches like Quantum Annealing (QA) \cite{morita2008mathematical} or QAOA \cite{farhi2014quantum} which promises speed-ups compared to classical algorithms \cite{farhi2016quantum}. To solve other optimization problems using these technologies they need to be transformed into a QUBO representation. Numerous well-known optimization problems such as boolean formula satisfiability (SAT), knapsack, graph coloring, the traveling salesman problem (TSP), or max clique have already been translated to QUBO form~\cite{8,12,13,15,lucas2014ising}. Note that Ising problems \cite{cipra1987introduction} are isomorphic to QUBO problems \cite{zick2015experimental}.

\subsection{QAOA}
The Quantum Approximate Optimization Algorithm (QAOA) is a hybrid quantum-classical algorithm proposed by Farhi et al. in 2014\cite{farhi2014quantum} for solving combinatorial optimization problems, which involves finding the best solution from a finite set of possible solutions. Let $C(x)$ be a cost function, where $x$ represents a binary string encoding a possible solution. The goal is to find the $x$ that minimizes $C(x)$. QAOA encodes this optimization problem into a quantum circuit, which can be parameterized by angles $\gamma$ and $\beta$. The quantum circuit prepares a quantum state $|\psi(\gamma, \beta)\rangle$ that represents a superposition of all possible solutions. The quantum circuit consists of alternating layers of two types of operators: the cost operator $U_C$ and the mixer operator $U_B$. The cost operator is responsible for encoding the cost function into the quantum state, while the mixer operator is responsible for exploring different solutions efficiently. The quantum state $|\psi(\gamma, \beta)\rangle$ prepared by the circuit is given by:
\[
|\psi(\gamma, \beta)\rangle = e^{-i\gamma_p U_B}e^{-i\beta_p U_C} \cdots e^{-i\gamma_1 U_B}e^{-i\beta_1 U_C}|+\rangle^{\otimes n}
\]
where $|+\rangle^{\otimes n}$ represents the initial state of $n$ qubits initialized to the superposition state, and $U_C$ and $U_B$ are the cost and mixer operators, respectively which are applied $p$ times. $p$ is a hyperparameter that needs to be manually specified. The parameters $\gamma$ and $\beta$ control the evolution of the quantum state.

The next step involves optimizing the parameters  $\gamma$ and $\beta$ to minimize the expectation value of the cost function. This optimization process is typically performed using classical optimization algorithms such as gradient descent or genetic algorithms. Given the quantum state $|\psi(\gamma, \beta)\rangle$, the expectation value of the cost function can be calculated as $E(\gamma, \beta) = \langle \psi(\gamma, \beta) | C | \psi(\gamma, \beta) \rangle
$. The goal is to find the optimal parameters $\gamma^*$ and $\beta^*$ that minimize $E(\gamma, \beta)$. This optimization process involves iteratively updating the parameters.

\subsection{Maximum Clique}
In graph theory, the Maximum Clique problem involves finding the largest subset of vertices $V' \subseteq V$ in a graph $G(V,E)$ such that every pair of vertices is connected by an edge. This problem has extensive applications across various domains, including social network analysis and bioinformatics~\cite{eblen2011maximum, rossi2015parallel}. To formulate the Maximum Clique problem as a QUBO problem, binary variables $x_i$ are used for each vertex $i$, where $x_i = 1$ indicates that vertex $i$ is included in the clique, and $x_i = 0$ otherwise. The Hamiltonian can therefore be written as:

\[
H(x) = \sum_i -x_i + A \cdot \sum_{(i,j) \in \overline{E}} x_i x_j 
\]

The second summand of $H$ enforces the solution to be a clique while the first summand rewards larger cliques \cite{lucas2014ising}.

\subsection{Hamilton Cycles}
Let $G(V,E)$ be a graph. The Hamilton Cycles problem asks if there is a path that starts from vertex $v_0$, visits every other vertice exactly once, and ends in vertex $v_0$ \cite{lucas2014ising}. This problem has practical applications in various fields, including logistics, transportation, and circuit design \cite{kawarabayashi2001survey, laporte2007locating}. To formulate this problem as a QUBO we introduce binary variables $x_{i,j}$ with $i \in [1..|V|]$ and $j \in [1..|V|]$. $x_{i,j} = 1$ iff vertex $i$ is at position $j$ of the cycle. The Hamiltonian can now be written as:

\begin{equation*}
\begin{aligned}
    H(x) &= \sum_i -x_i + A \cdot \sum_{i,j} \sum_{k,l} x_{i,j} x_{k,l} \cdot I[i=k \lor j=l \\ &\lor (l = j + 1 \land (i,k) \notin E) \lor (l = |V| - 1 \land j = 0 \\ &\land (i,k) \notin E)]
\end{aligned}
\end{equation*}

$H$ consists of three constraints: (1) each vertex must be visited (2) two vertices can't be at the same position in the cycle (3) two vertices can not be in neighboring positions of the cycle if there is no edge in the graph connecting them.

\subsection{Graph Coloring}
The Graph Coloring problem encompasses a wide range of applications from scheduling to register allocation in compilers, and even to radio frequency assignment in wireless communication networks \cite{ahmed2012applications}. At its core, the problem revolves around assigning colors to the vertices of a graph in such a way that no two adjacent vertices share the same color. Let $G=(V,E)$ be a graph, and $K$ be the number of available colors. To formulate this problem as a QUBO we introduce binary variables $x_{i,k}$ representing the assignment of color $k$ to vertex $i$ \cite{lucas2014ising}.

\begin{equation*}
\begin{aligned}
    H(x) &= \sum_{i,k} -x_{i,k} + A \cdot \sum_{i,k_1} \sum_{j,k_2} x_{i,k_1} x_{j,k_2} \cdot I[i=j \\ &\lor (k_1 = k_2 \land (i,j) \in E)]
\end{aligned}
\end{equation*}

$H$ encodes the two constraints that each vertex can only have one color and two adjacent vertices can't have the same color.

\subsection{Vertex Cover}

Given an undirected graph $G=(V,E)$, the task is to determine the minimum number of vertices that need to be selected such that each edge in the graph is incident to at least one selected vertex. This problem is known to be NP-hard, and its decision form is classified as NP-complete \cite{karp1975computational}. To formulate this problem as a QUBO problem, we introduce binary variables $x_v$ for each vertex $v$, where $x_v = 1$ iff the vertex is selected and $x_v = 0$ iff it is not selected \cite{lucas2014ising}. The Hamiltonian $H(x) = H_A(x) + H_B(x)$ for this problem is an addition of $H_A(x)$ and $H_B(x)$, where $H_A(x)$ represents the constraint ensuring that every edge has at least one selected vertex. We encode the constraint using $H_A(x)$ as follows:
\[
H_A(x) = A \cdot \sum_{(u,v) \in E} (1 - x_u)(1 - x_v)
\]
$H_B(x)$ is used to minimize the number of selected vertices:
\[
H_B(x) = \sum_v x_v
\]

\subsection{Graph Isomorphism}
Graph Isomorphism (GI) is an important problem in graph theory that asks whether two graphs are structurally equivalent, albeit possibly differing in their vertex and edge labels. Formally, two graphs $G_1=(V_1, E_1)$ and $G_2=(V_2, E_2)$ are considered isomorphic if there exists a bijective mapping between their vertices such that their edge structures remain unchanged. In contrast to Maximum Clique, Hamilton Cycles, Graph Coloring and Vertex Cover, the complexity class for GI is still unknown (although it is expected to be in NP-intermediate) \cite{npintermediate}.

To formulate GI as a QUBO problem, we introduce binary variables $x_{i,j}$ representing the mapping of vertex $i$ of $G_1$ to vertex $j$ of $G_2$. The Hamiltonian can now be formulated as \cite{lucas2014ising}:

\begin{equation*}
\begin{aligned}
    H(x) &= \sum_i -x_i + A \cdot \sum_{i_1, j_1} \sum_{i_2, j_2} x_{i_1, j_1} x_{i_2, j_2} \cdot \\ &I[i_1 = j_2 \lor j_1 = j_2 \lor ((i_1, i_2) \in E_1 \land (j_1, j_2) \notin E_2) \\ &\lor ((i_1, i_2) \notin E_1 \land (j_1, j_2) \in E_2)]
\end{aligned}
\end{equation*}
\ \\

\section{Related Work}
In this paper, we propose the concept of \textit{Semi-Symmetries} in QUBO matrices $Q$ and an algorithm for factoring them out into ancilla qubits to reduce the number of couplings and therefore the number of CNOT gates and the circuit depth of QAOA. There are already two well-known types of symmetries in QUBO matrices: \textit{bit-flip-symmetry} and \textit{qubit-permutation-symmetry} \cite{shaydulin2021error, shaydulin2021exploiting, shaydulin2020classical}. Symmetry is defined here regarding the solution vectors $\{x\}$ and their associated energies $\{x^TQx\}$.

\subsection{Bit-flip-symmetry}

\textit{Bit-flip-symmetry} denotes the property of QUBOs that the inverse bit vector $x_I = 1 - x$ to a bit vector $x$ both have the same energy: ${(x_I)}^TQx_I = x^TQx$. \textit{Bit-flip-symmetries} occur, for example, in the Max-Cut problem:
\[
H(x) = \sum_{(i,j) \in E} - x_i - x_j + 2 x_i x_j
\]
\ \\
\textit{Bit-flip-symmetry} can be identified in a QUBO matrix $Q$ by substituting $x_i \gets (1 - x_i)$ and $x_j \gets (1 - x_j)$:
\begin{equation*}
\begin{aligned}
    H(x) &= \sum_{(i,j) \in E} - (1 - x_i) - (1 - x_j) + 2 (1 - x_i) (1 - x_j) = \\ &= \sum_{(i,j) \in E} -2 + x_i + x_j + 2 (1 - x_j - x_i + x_i x_j) = \\ &= \sum_{(i,j) \in E} - x_i - x_j + 2 x_i x_j
\end{aligned}
\end{equation*}
\ \\
Since the energy stays the same the QUBO contains a bit-flip-symmetry. Eliminating \textit{bit-flip-symmetry} can be done by removing the last qubit and assigning it the value $0$. Then, the remaining $(n-1) \times (n-1)$ QUBO matrix still encodes the original Hamiltonian.

\subsection{Qubit-permutation-symmetry}
Qubits $i$ and $j$ are \textit{qubit-permutation-symmetrical} if they have the same coupling values to all other qubits, i.e.:
\[
\forall \: k \in [1..n] : Q_{i,k} = Q_{j,k}
\]
This implies that for all $x^{(i=1, j=0)}$ it holds:
\[ H(x^{(i=1, j=0)}) = H(x^{(i=0, j=1)}) \] We use the notation $x^{(i=1, j=0)}$ for an arbitrary solution vector $x$ with qubit $i$ having value $1$ and qubit $j$ having value $0$. However, a trivial reduction of such a QUBO is not possible, since there are $3$ cases that have different energies: $x^{(i=0, j=0)}$, $x^{(i=1, j=0)}$ and $x^{(i=1, j=1)}$.

\subsection{Choosing a value for $p$}
Several works \cite{niu2019optimizing, pan2022automatic, ni2023more} have analyzed the influence of circuit depth on the success probability for QAOA. Note that the term \textit{depth} is sometimes used synonymously with the number of layers, which we refer to as $p$. In this paper, we exclusively refer to \textit{depth} as the depth of the transpiled quantum circuit. To select the optimal number of repetitions \textit{p}, several approaches have been proposed for automatically setting this hyperparameter \cite{pan2022automatic, ni2023more, pan2022efficient, lee2021parameters}. In our experiments, we always tested $p \in \{1,2,3\}$.

\subsection{Other approaches for eliminating couplings in $Q$}

In \textit{Algorithm 1}, the original $Q$ is modified by factoring out \textit{semi-symmetries} into additional ancilla qubits. However, we show that in doing so, the energy landscape for valid solutions is not altered. In contrast, there are heuristic approaches that alter the energy landscape to simplify $Q$. For example, in the paper \cite{sax2020approximate}, an approach was introduced to reduce the number of couplings in a QUBO by simply setting the smallest couplings to $0$ since they have the smallest influence on the energy landscape. By altering the energy landscape in this manner, it can no longer be guaranteed that the optimal solution $x_{mod}^*$ of the modified QUBO $Q_{mod}$ corresponds to the optimal solution $x^*$ of the original QUBO $Q$. 

Ising graphs associated to real-world problems, such as Airport Traffic Graphs, often exhibit a power-law structure \cite{ayanzadeh2023frozenqubits}, where some nodes have many more connections than others. In the paper \cite{ayanzadeh2023frozenqubits}, an approach is presented on how to partition the graphs with respect to these 'hubs'. This eliminates many couplings of the Hamiltonian, and the individual subgraphs can then be solved individually using a divide-and-conquer approach. A detailed analysis of the performance of QAOA depending on the graph structure is provided in \cite{herrman2021impact}. In \cite{ponce2023graph}, an approach is proposed on how large Max-Cut QUBOs can be solved by decomposing them into many smaller QUBOs. A similar approach is pursued in \cite{majumdar2021depth}.

There are already several papers \cite{shaydulin2021error, shaydulin2021exploiting, shaydulin2020classical} that exploit symmetries in QUBOs to generate more efficient and shorter QAOA circuits. In \cite{shaydulin2021error}, a method is proposed for leveraging \textit{bit-flip-symmetry} and \textit{qubit-permutation-symmetry} on Max-Cut graphs. In \cite{shaydulin2020classical} various types of symmetries that are relevant to QAOA and classical optimization problems are discussed. One prominent type is variable (qubit) permutation symmetries, which are transformations that rearrange the qubits of the quantum state without changing the problem's objective function. Such a symmetry can be caused when a graph contains automorphisms (a mapping of the graph to itself). The authors show that if a group of variable permutations leaves the objective function invariant, then the output probabilities of QAOA will be the same across all bit strings connected by such permutations, regardless of the chosen QAOA parameters and depth which can be used to reduce the dimension of the effective Hilbert space.

\section{Algorithm}

We start this section by providing a formal definition of \textit{conflicting qubits} and \textit{semi-symmetries}.

\begin{definition}[Conflicting qubits]
Let $H(x) = x^T Q x$ be the energy of a solution $x$. Qubits $i$ and $j$ are called \textit{conflicting}, iff for every solution $x^{(i=1, j=1)}$ it holds that: \[ H(x^{(i=1, j=1)}) > \{ H(x^{(i=1, j=0)}), H(x^{(i=0, j=1)}), H(x^{(i=0, j=0)}) \} \].
\end{definition}

\begin{algorithm}[H]
   \caption{Factoring our Semi-Symmetries}
   \label{alg:example}
\begin{algorithmic}
   \State {\bfseries Input:} QUBO matrix $Q$ of size $n \times n$
   \State \:\:\:\:\:\:\:\:\:\:\:\:\: $numAncillas \in \mathbb{N}$
   \State \:\:\:\:\:\:\:\:\:\:\:\:\: $z \in \mathbb{R^+}$
   \State $n_{new} = n$ 
   \State $cL \gets$ \textit{$getConflictList(Q,n_{new})$}
   \While{$len(cL) > 0$}
   \State $syms, (i,j) = getMostSymQubits(Q, n_{new}, cL)$
   \If{$len(syms) < 3 $ \textbf{or} $n_{new} = n + numAncillas$}
   \State \textbf{break}
   \EndIf
   \State $n_{new} = n_{new} + 1$
   \State $Q \gets enhance(Q, n_{new}, (i,j), syms)$
   \State $cL \gets$ \textit{$getConflictList(Q,n_{new})$}
   \EndWhile \\
   \Return $Q$
   
   \State 
   \Function{getConflictList}{$Q,n$}
   \State $cL \gets$ []
   \State $Z = [sum([Q_{i,j} : j \in [1..n] \land Q_{i,j} < 0]) : i \in [1..n]]$
   \For{$i \in [1..n], j \in [1..n]$}
   \If{$i < j \land Q_{i,j} > -Z[i]-Z[j]$}
   \State $cL$.append($(i,j)$)
   \EndIf
   \EndFor
   \State \textbf{return} $cL$
   \EndFunction
   
   \State 
   \Function{getMostSymQubits}{$Q, n, cL$}
   \State $best \gets (0,1)$
   \State $bestSyms \gets$ []
   \For{$(i,j) \in cL$}
   \State $syms = [k : k \in [1..n] : Q_{i,k} = Q_{j,k} \neq 0]$
   \If{len($syms$) $\geq$ len($bestSyms$)}
   \State $best \gets (i,j)$
   \State $bestSyms \gets syms$
   \EndIf
   \EndFor
   \State \textbf{return} $bestSyms, best$
   \EndFunction
   
   \State 
   \Function{enhance}{$Q, n, (i,j), syms$}
   \State $Q_{i,i} = Q_{i,i} + z$
   \State $Q_{j,j} = Q_{j,j} + z$
   \State $Q_{n,n} = z$
   \State $Q_{i,n} = -2 \cdot z$
   \State $Q_{j,n} = -2 \cdot z$
   \State $Q_{i,j} = 2 \cdot z$
   \For{$k \in syms$}
   \State $Q_{k,n} = Q_{i,k}$
   \State $Q_{i,k} = 0$
   \State $Q_{j,k} = 0$
   \EndFor
   \State \textbf{return} $Q$
   \EndFunction
   
\end{algorithmic}
\end{algorithm}

\begin{definition}[Semi-symmetry]
Conflicting qubits $(i,j)$ are semi-symmetric if and only if:
\[
\exists \: U \subseteq \{1..n\} \backslash \{i,j\} \land |U| \geq 3 : \forall \: k \in U : Q_{i,k} = Q_{j,k} \neq 0
\]
In other words, two \textit{conflicting qubits} $(i,j)$ are \textit{semi-symmetric} iff there are at least $3$ other qubits to which $i$ and $j$ have the same non-zero couplings. This is a weakened definition of symmetry compared to \textit{qubit-permutation-symmetry} where qubits $i$ and $j$ needed the same couplings to \textit{all} other qubits.
\end{definition}

Having defined \textit{semi-symmetries}, we now formulate our algorithm to factor them out into additional ancilla qubits. The complete procedure is summarized in \textbf{Algorithm 1} which takes as input the QUBO matrix $Q$, an integer \textit{numAncillas} which represents the maximal number of ancilla qubits we want to use and a positive penalty value $z$.

The algorithm starts with identifying all \textit{conflicting qubits}, calculating the \textit{semi-symmetries} for all of them, and selecting the qubits $(i,j)$ with the largest \textit{semi-symmetry}. Then it enhances the QUBO matrix by appending an ancilla qubit $a$ and factoring out the semi-symmetries. To do so, the algorithm copies all coupling values between qubit $i$ and the list of symmetrical qubits $syms$ to the coupling values between $a$ and $syms$, removes all coupling values between $i$ and $syms$ as well as between $j$ and $syms$. Further, it adds a constraint that enforces the ancilla qubit $a$ to be $1$ if and only if either qubit $i$ or $j$ is $1$.

The algorithm has runtime complexity $O(n^3)$ since the while-loop runs for at most $n$ iterations and the functions \textit{getMostSymQubits()} and \textit{getConflictList()} both have runtime $O(n^2)$ and the function \textit{enhance()} has constant runtime.

\subsection{Proof-of-Concept Example}
In the following section, we demonstrate our algorithm for a simple proof-of-concept example. To do this, we consider the following graph:

\begin{figure}[H]
 \centering
  \includegraphics[scale=0.3]{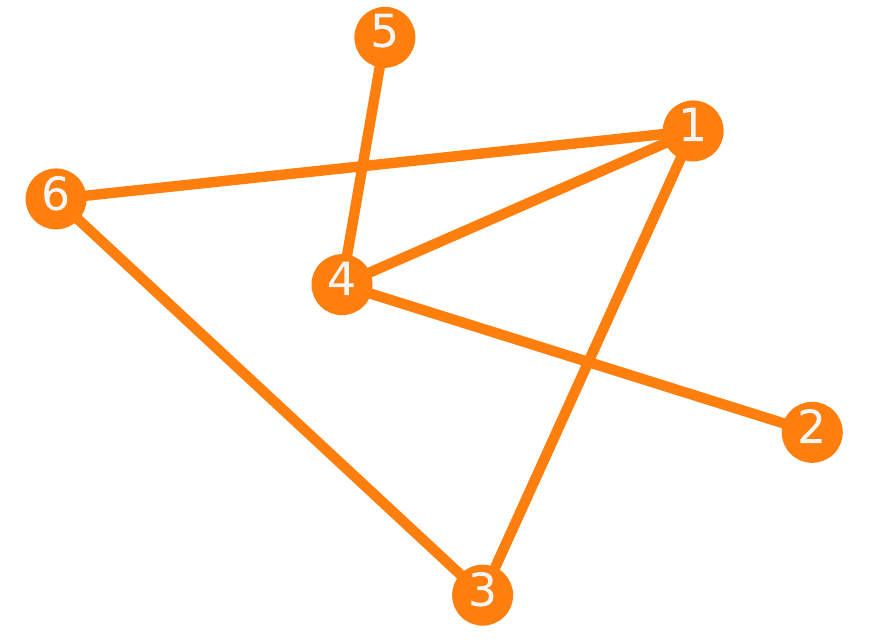}
  \caption{A simple graph for demonstrating proof-of-concept.}\label{fig:main_charts}
\end{figure}

We now want to find the largest clique (Maximum Clique) for the graph $G=(V,E)$ in Figure 1, i.e. the largest set of nodes for which each pair of nodes is connected by an edge. The Hamiltonian that encodes this problem is given by :

\[
H = \sum_i -x_i + \sum_{(i,j) \in \overline{E}} 3 x_i x_j
\]

The QUBO matrix $Q$ for Maximum Clique and the graph from \textbf{Figure 1} is listed in \textbf{Table I (left)}. It requires $6$ qubits and $9$ couplings. $Q$ contains a \textit{semi-symmetry} between qubits $2$ and $5$ which can be factored out into an additional ancilla qubit $7$ (see \textbf{Table I (right)}). The modified QUBO matrix $Q_{mod}$ requires $7$ qubits but only $8$ couplings.

\begin{figure*}[t!]
\centering
\minipage{0.88\textwidth}
  \centering
  \includegraphics[width=\linewidth]{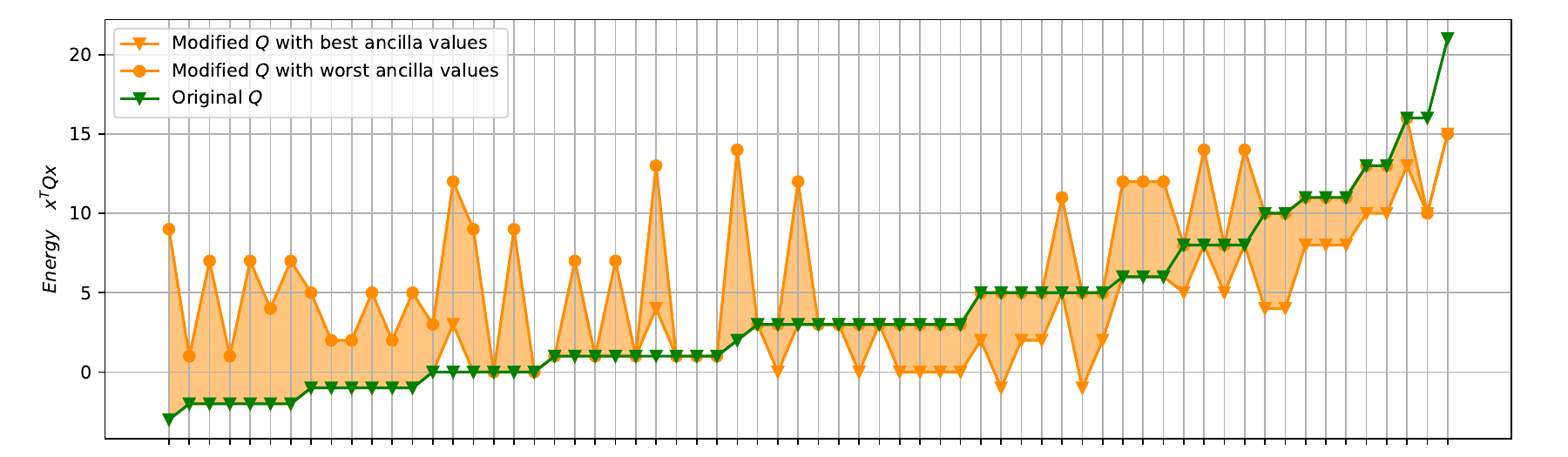}
  \includegraphics[width=\linewidth]{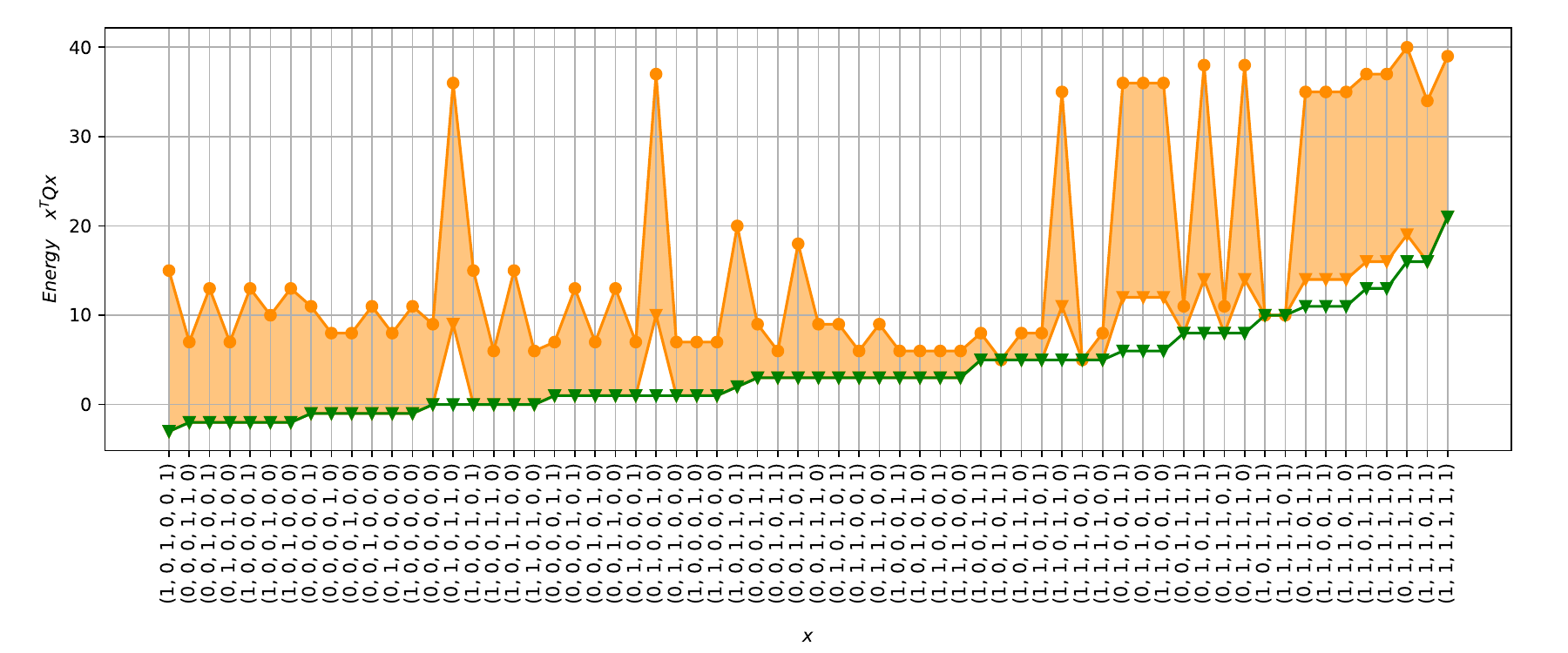}
  \caption{The green line represents the sorted energy spectrum of the left QUBO in Table I. The orange lines represent the energy spectrum of the right QUBO in Table I with the lower graph in both plots being the energetically more favorable choice of the ancilla value while the upper graph in both plots representing the energetically less favorable choice. The upper plot represents the energy spectrum of $Q_{mod}$ with $z=3$ and the lower plot with $z=9$. For $z=3$ we can see that invalid solutions can have a lower energy in $Q_{mod}$ than in $Q$ but if we increase $z$ all invalid solutions have an energy equal or higher than in $Q$. But even $z=3$ is already sufficient for the global optimum $x^*$ in $Q$ to also be the global optimum in $Q_{mod}$.}\label{fig:main_charts}
\endminipage
\end{figure*}

\begin{table}[H]
  \centering
    \begin{tabular}{|l|l|l|l|l|l|} \hline 
         -1& 3& & & 3& \\\hline
         & -1& 3& & 3& 3\\\hline
         & & -1& 3& 3& \\\hline
         & & & -1& & 3\\\hline
         & & & & -1& 3\\\hline
         & & & & & -1\\\hline
    \end{tabular}
    \vspace{1em}
    \vspace{1em}
    \begin{tabular}{|l|l|l|l|l|l|l|} \hline 
         -1& & & & & & 3\\\hline
         & 2& & & 9& & -6\\\hline
         & & -1& 3& & & 3\\\hline
         & & & -1& & 3& \\\hline
         & & & & 2& & -6\\\hline
         & & & & & -1& 3\\\hline
         & & & & & & 3\\\hline
    \end{tabular}
    \caption{(Left) QUBO matrix $Q$ for Maximum Clique and the graph in Figure 1. (Right) Modified QUBO matrix $Q_{mod}$ of $Q$ using \textbf{Algorithm 1}. The \textit{semi-symmetry} between qubits $2$ and $5$ was factored out into an additional ancilla qubit.}
\end{table}

In the following section, we theoretically show that our algorithm doesn't change the energy landscape for valid solutions and in section \textit{E.} we analyze the energy spectra for both QUBO matrices in \textit{Table I}.

\subsection{Theoretical Analysis for Correctness}
We can prove that our modified QUBO $Q_{mod}$ has the same optimal solutions as $Q$ with the best choice of ancilla values, i.e. \textbf{Algorithm 1} doesn't change the energies of valid solutions and doesn't decrease the energy of invalid solutions. Valid solutions $x$ are bit-vectors that don't violate \textit{conflicting qubit} constraints, i.e. if $(i,j)$ are conflicting then $x_i = 0$ or $x_j = 0$. Invalid solutions are bit-vectors with $x_i = 1$ and $x_j = 1$. 

\begin{proposition}
If we choose $z = \sum_{(i,j)} |Q_{i,j}|$, valid solutions $x$ have the same energy regarding $Q$ as to $Q_{mod}$ with the best values for the ancilla qubits $x_{mod} = x + [x_a]$. The energy of invalid solution doesn't decrease with respect to $Q_{mod}$ even with the best ancilla values.
\end{proposition}
\begin{proof}
Let $Q$ be any QUBO matrix and $x \in \mathbb{B}^n$ be any solution vector. The energy $E$ for $x$ corresponds to $E = x^T Q x$. Let $(x_i,x_j)$ be a pair of conflicting qubits, i.e. no valid solution $x$ contains assignments $i=1$ and $j=1$ at the same time. Further assume that $(x_i,x_j)$ are semi-symmetrical and \textbf{Algorithm 1} factored out the semi-symmetries into an ancilla qubit $x_a$.

\ \\
\textbf{Case 1: $x_i = 0, x_j = 0, x_a = 0$}: in this case, we can easily see that the energy of $x_{mod} = x + [0]$ regarding $Q_{mod}$ is identical to the original energy: $E_{mod} = {(x + [0])}^T \cdot Q_{mod} \cdot (x + [0]) = E$.

\ \\
\textbf{Case 2: $x_i = 0, x_j = 0, x_a = 1$}: in this case the modified energy corresponds to: $E_{mod} = E + z + \sum_{k \in syms} Q_{i,k}$. Since we can choose $z = \sum_{(i,j)} |Q_{i,j}|$, it holds that $z + \sum_{k \in syms} Q_{i,k} \geq 0$. Therefore: $E_{mod} \geq E$.

\ \\
\textbf{Case 3: $x_i = 1, x_j = 0, x_a = 0$}: $E_{mod} = E + z - \sum_{k \in syms} Q_{i,k}$. Again, since $z = \sum_{(i,j)} |Q_{i,j}|$, it holds that $z - \sum_{k \in syms} Q_{i,k} \geq 0$. Therefore: $E_{mod} \geq E$.

\ \\
\textbf{Case 4: $x_i = 1, x_j = 0, x_a = 1$}: $E_{mod} = E + z + z - 2z + \sum_{k \in syms} Q_{i,k} - \sum_{k \in syms} Q_{i,k} = E$.

\ \\
\textbf{Case 5: $x_i = 0, x_j = 1, x_a = 0$}: analogous to case 3.

\ \\
\textbf{Case 6: $x_i = 0, x_j = 1, x_a = 0$}: analogous to case 4.

\ \\
\textbf{Case 7: $x_i = 1, x_j = 1, x_a = 0$}: $E_{mod} = E + z + z + 2z - \sum_{k \in syms} Q_{i,k} - \sum_{k \in syms} Q_{i,k} > E$.

\ \\
\textbf{Case 8: $x_i = 1, x_j = 1, x_a = 1$}: $E_{mod} = E + z + z + z - 2z - 2z + 2z - \sum_{k \in syms} Q_{i,k} - \sum_{k \in syms} Q_{i,k} + \sum_{k \in syms} Q_{i,k} = E + z - \sum_{k \in syms} Q_{i,k} \geq E$.

\ \\
The best choices for the ancilla qubit for valid solutions are \textit{case 1}, \textit{case 4} and \textit{case 6} which all have energy $E$. Therefore, the energy did not change for valid solutions. For invalid solutions (\textit{cases 7} and \textit{8}) the energy does not decrease.

\end{proof}

\subsection{Empirical Evaluation of the Energy Landscape for the PoC}
We now empirically investigate this theoretical finding in our proof-of-concept example. Since there are $6$ qubits, there are $2^6=64$ possible solutions $x$. For each $x$ we calculated the energy regarding $Q$ (Table I, left), see green lines in Figure 2. Further, we have calculated the energy in the modified QUBO (Table I, right) with both possible values (0 and 1) for the ancilla qubit 7. Then we have plotted for each $x$ the original energy, the energy in the modified $Q_{mod}$ with the worse choice for the ancilla qubit and the better choice for the ancilla qubit.

The upper plot in Figure 2 shows the result with $z=3$, and the lower plot shows the result with $z=9$. We can verify the proposition if we choose $z$ big enough, but often a lower value for $z$ is already enough for the original optimal solution $x$ to also be the optimal solution in $Q_{mod}$.

\section{Experiments}

To evaluate our approach, we selected the following five optimization problems: Maximum Clique, Hamilton Cycles, Graph Coloring, Vertex Cover and Graph Isomorphism.
\begin{figure*}
\centering
\minipage{0.97\textwidth}
  \centering
  \includegraphics[width=\linewidth]{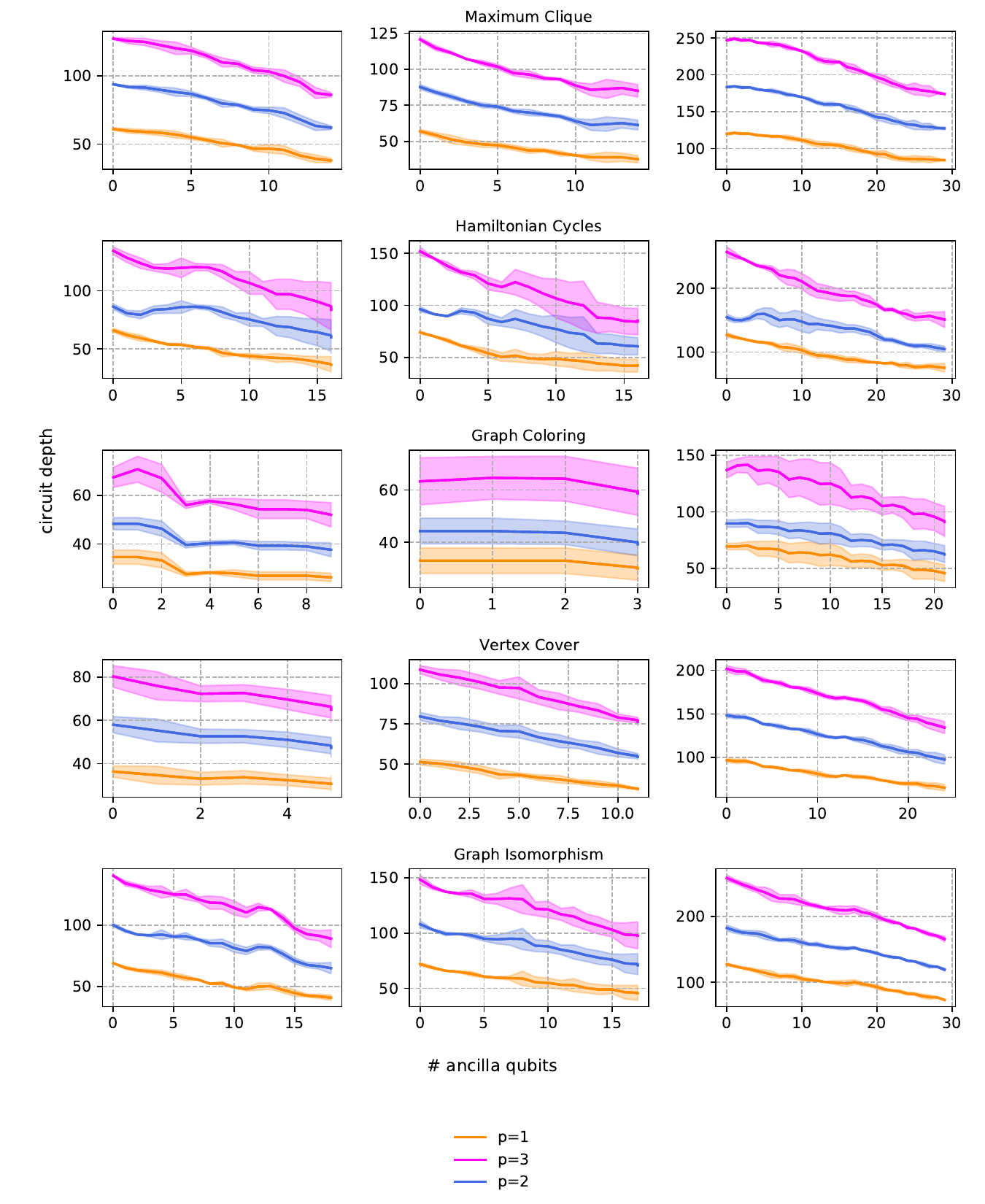}
  \caption{This plot shows the results of our experiments on reducing the depth of QAOA circuits by factoring out \textit{semi-symmetries} using the algorithm presented in section 3. We tested our approach for five optimization problems, including Maximum Clique, Hamilton Cycles, Graph Coloring, Vertex Cover, and Graph Isomorphism. The $x$-axis delineates the number of ancilla qubits $numAncillas$ employed, ranging from $0$ to $29$, representing the number of semi-symmetries the algorithm removed from the QUBO. Meanwhile, the $y$-axis captures the resultant circuit depth, offering insights into the efficiency gains achieved through our algorithm. Notably, the plots demonstrate a substantial reduction in circuit depth by up to $41\%$, depending on the specific problem and parameter settings}\label{fig:main_charts}
\endminipage
\end{figure*}
Each problem is parameterized by a set of problem parameters: for all of them this is the number of vertices $|V|$ and the number of edges $|E|$ and for Graph Coloring additionally the number of colors $K$. For each problem, we have examined $3$ different settings for these problem parameters (details about them are listed in appendix Table II). For each of the $5$ problems and each of the $3$ settings, we sampled $4$ random graph $G=(V,E)$. We then created the corresponding QUBO matrix $Q$ and used \textbf{Algorithm 1} to remove $numAncillas$ semi-symmetries which resulted in a modified QUBO matrix $Q_{mod}^{numAncillas}$. For each $numAncillas \in [1..29]$. We then counted how many couplings the QUBO $Q_{mod}^{numAncillas}$ contained. Figure 3 shows the results of this experiment. Here, the $x$-axis represents the number of ancilla qubits $numAncillas$ (the number of semi-symmetries our algorithm removed from the QUBO), and the $y$-axis represents the number of couplings of the QUBO $Q_{mod}^{numAncillas}$. The solid lines represent the mean over all 4 graphs and the shaded area represents the standard deviation. The algorithm was given a budget of $29$ ancilla qubits, however less qubits were used if the QUBOs contained less than $29$ semi-symmetries. The results show that our approach reduced the number of non-zero couplings in $Q_{mod}^{numAncillas}$ compared to $Q$ by up to $49\%$.

\begin{figure}[H]
 \centering
  \includegraphics[scale=0.39]{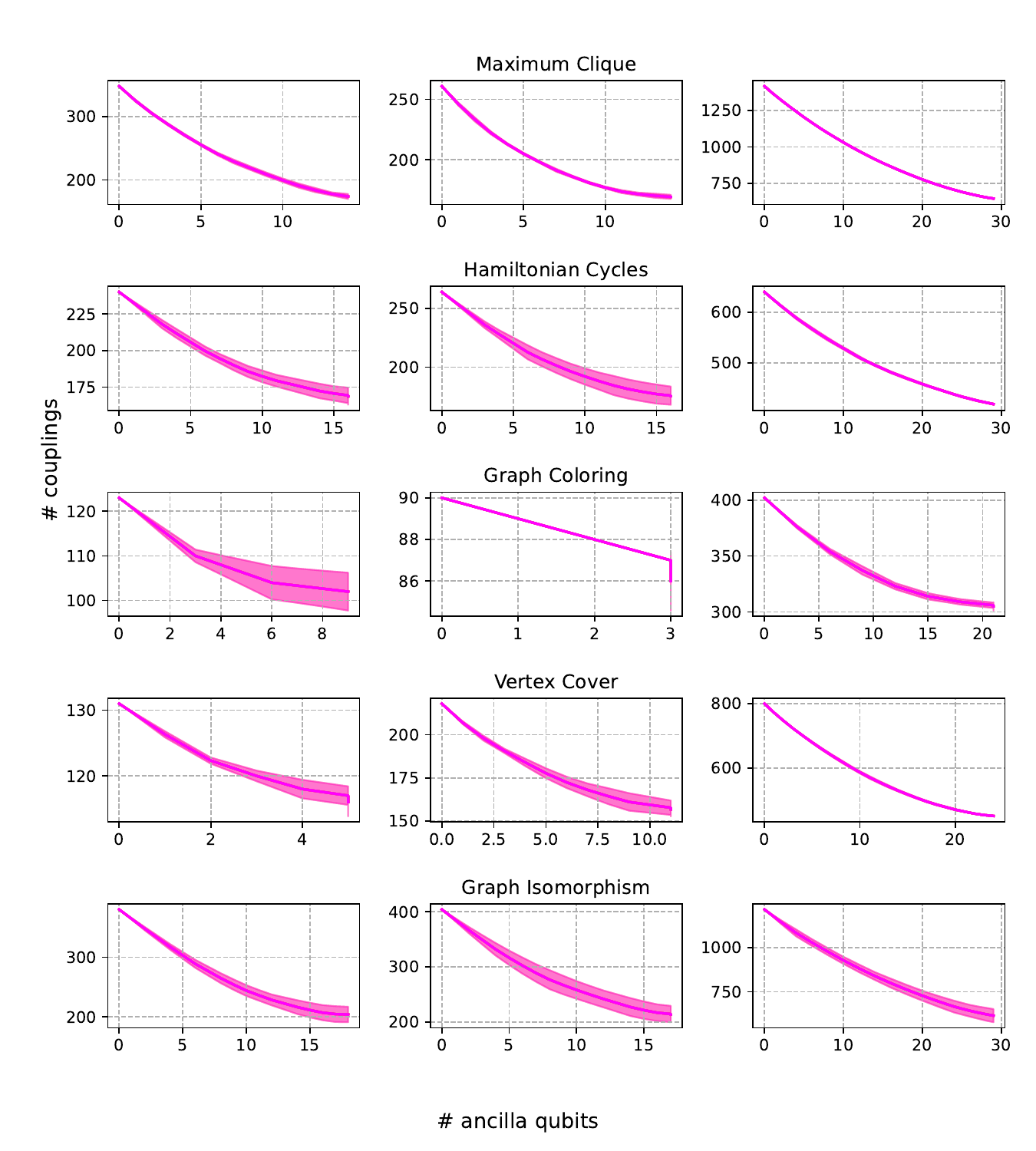}
  \caption{Pareto Front of \textit{$\#$ Ancilla Qubits} and \textit{$\#$ Couplings} for $5$ different optimization problems and $3$ problem parameter settings each.}\label{fig:main_charts}
\end{figure}

Based on the results of this first experiment, we transpiled the QAOA circuits for each QUBO $Q_{mod}^{numAncillas}$ using $p \in \{1,2,3\}$. The plots in Figure 4 visualize the relationship between the number of ancilla qubits ($x$-axis) and the circuit depth ($y$-axis). The results show that we were able to reduce the circuit depth, depending on the problem and the setting by up to $41\%$. For example in the \textit{Hamilton Cycles} problem in the last setting $(V=8, E=16)$ we reduced the circuit depth for $p=3$ from $257$ to $151$ $(-41\%)$ using $29$ ancilla qubits. In the first setting $(V=30, E=87)$ of the \textit{Maximum Clique} problem, our algorithm reduced the circuit depth from $127$ to $86$ $(-32\%)$ using just $16$ ancilla qubits (for $p=3$).

In our experiments, \textit{Graph Coloring} contained the least number of \textit{semi-symmetries}, here we could only reduce the circuit depth from $137$ to $91$ $(-33.5\%)$ for the largest problem setting and $p=3$. We appended the code as supplementary material for all of our experiments for easy reproduction.

\section{Conclusion}
This paper presents a novel approach to address the challenges associated with the practical implementation of the QAOA in tackling combinatorial optimization problems on quantum computers. By factoring out \textit{semi-symmetries} into additional ancilla qubits, we have successfully reduced the number of non-zero couplings and therefore also the number of CNOT gates and the depth of the QAOA circuits, thereby enhancing their scalability. We defined \textit{semi-symmetries} as a set of at least $3$ identical, non-zero couplings of two \textit{conflicting} qubits to other qubits.

We have shown theoretically that our approach does not change the energy landscape for valid solutions, i.e. solutions that don't violate \textit{conflicting qubits}.

Through experimental validation on five well-known optimization problems, including Maximum Clique, Hamilton Cycles, Graph Coloring, Vertex Cover, and Graph Isomorphism, we have shown that our approach reduced the circuit depth by up to $41\%$ and reduced the number of coupling in the QUBO by up to $49\%$.

Looking ahead, further research can investigate the robustness of our approach against noise and error rates inherent in near-term quantum devices.

\section*{Acknowledgment}
This publication was created as part of the Q-Grid project (13N16179) under the ``quantum technologies -- from basic research to market'' funding program, supported by the German Federal Ministry of Education and Research.

\bibliographystyle{plain}
\bibliography{bibliography}

\appendix

\begin{table}[H]
  \centering
    \begin{tabular}{|l|l|l|} \hline
         $V=30, E=87$ & $V=30, E=174$ & $V=60, E=354$ \\\hline
         $V=6, E=10$ & $V=6, E=8$ & $V=8, E=16$ \\\hline
         $V=10, E=31$ & $V=10, E=20$ & $V=20, E=114$ \\
         $K=3$ & $K=3$ & $K=3$ \\\hline
         $V=30, E=131$ & $V=30, E=218$ & $V=50, E=800$ \\\hline
         $V=6, E=10$ & $V=6, E=8$ & $V=8, E=16$ \\\hline
    \end{tabular}
    \vspace{1em}
    \caption{Problem setting values for the experiments. Using these values graphes were uniformly sampled.}
\end{table}

\end{document}